\documentclass[12pt]{amsart}
\usepackage{amsmath}
\usepackage{amsthm}
\usepackage{amssymb}
\usepackage{amscd}
\usepackage{bbold}
\usepackage[pdftex]{graphicx}
\usepackage{enumitem}
\usepackage{subcaption}
\usepackage{float}

\usepackage{bbm}
\usepackage[utf8]{inputenc}
\usepackage[english]{babel}
\usepackage{mathrsfs}

\usepackage{appendix}

\newtheorem{lemma}{Lemma}[section]
\newtheorem{theorem}[lemma]{Theorem}
\newtheorem{proposition}[lemma]{Proposition}

\theoremstyle{definition}
\newtheorem{definition}[lemma]{Definition}

\newtheorem{remark}[lemma]{Remark}

\DeclareMathOperator*{\argmin}{arg\,min}

\newcommand{\CC}{{\mathbb C}}

\newcommand{\NN}{{\mathbb N}}

\newcommand{\RR}{{\mathbb R}}

\newcommand{\ZZ}{{\mathbb Z}}

\newcommand{\Hcal}{\mathcal{H}}

\newcommand{\Scal}{\mathcal{S}}

\def\benm{\begin{enumerate}}            
\def\eenm{\end{enumerate}}              

\title[Adapted Decimation for Arbitrary Sigma-Delta Quantization]{Adapted Decimation on Finite Frames for Arbitrary Orders of Sigma-Delta Quantization}

\date{}



\author{Kung-Ching Lin}
\address{Norbert Wiener Center\\
         Department of Mathematics \\
         University of Maryland \\
         College Park, MD 20742 \\
         USA, \\
         Tel. +1 301-405-5058; Fax. +1 301-314-0827}
\email{kclin@math.umd.edu}

\begin{document}

\newcounter{bean}

\begin{abstract}
In Analog-to-digital (A/D) conversion, signal decimation has been proven to greatly improve the efficiency of data storage while maintaining high accuracy. When one couples signal decimation with the $\Sigma\Delta$ quantization scheme, the reconstruction error decays exponentially with respect to the bit-rate. We build on our previous result, which extended signal decimation to finite frames, albeit only up to the second order. In this study, we introduce a new scheme called adapted decimation, which yields polynomial reconstruction error decay rate of arbitrary order with respect to the oversampling ratio, and exponential with respect to the bit-rate. 
\end{abstract}

\maketitle


\section{Introduction}\label{intro}

\subsection{Background}\label{background}
	Analog-to-digital (A/D) conversion is a process where bandlimited signals, e.g., audio signals, are digitized for storage and transmission, which is feasible thanks to the classical sampling theorem. In particular, the theorem indicates that discrete sampling is sufficient to capture all features of a given bandlimited signal, provided that the sampling rate is higher than the Nyquist rate.\par
	
	Given a function $f\in L^1(\mathbb{R})$, its Fourier transform $\hat{f}$ is defined as
\[
	\hat{f}(\gamma)=\int_{-\infty}^{\infty} f(t)e^{-2\pi\imath t\gamma}\, dt.
\]
	The Fourier transform can also be uniquely extended to $L^2(\mathbb{R})$ as a unitary transformation.\par
\begin{definition}
	Given $f\in L^2(\mathbb{R})$, $f\in PW_{[-\Omega,\Omega]}$ if its Fourier transform $\hat{f}\in L^2(\mathbb{R})$ is supported in $[-\Omega,\Omega]$. 
\end{definition}
	An important component of A/D conversion is the following theorem:
\begin{theorem}[Classical Sampling Theorem]
	Given $f\in PW_{[-1/2,1/2]}$, for any $g\in L^2(\RR)$ satisfying 
\begin{itemize}
\item	$\hat{g}(\omega)=1$ on $[-1/2, 1/2]$
\item	$\hat{g}(\omega)=0$ for $|\omega|\geq 1/2+\epsilon$,
\end{itemize}
	 with $\epsilon>0$ and $T\in(0,1-2\epsilon)$, $t\in\mathbb{R}$, one has\\
\[
	f(t)=T\sum_{n\in\mathbb{Z}}f(nT)g(t-nT),
\]
	where the convergence is both uniform on compact sets of $\RR$ and in $L^2$.
\end{theorem}
	As an extreme case, for $g(t)=\sin(\pi t)/(\pi t)$ and $T=1$, the following identity holds in $L^2(\RR)$:
\[
	f(t)=\sum f(n)\frac{\sin(\pi(t-n))}{\pi(t-n)}.
\]
	However, the discrete nature of digital data storage makes it impossible to store exactly the samples $\{f(nT)\}_{n\in\mathbb{Z}}$. Instead, the quantized samples $\{q_n\}_{n\in\mathbb{Z}}$ chosen from a pre-determined finite alphabet $\mathscr{A}$ are stored. This results in the following reconstructed signal
\[
	\tilde{f}(t)=T \sum q_n g(t-nT).
\]

	As for the choice of the quantized samples $\{q_n\}_n$, we shall discuss the following two schemes
\begin{itemize}
\item	Pulse Code Modulation (PCM):\par
	Quantized samples are taken as the direct-roundoff of the current sample, i.e.,
\[
	q_n=Q_0(f(nT)):=\argmin_{q\in\mathscr{A}}|q-f(nT)|.
\]
\item	$\Sigma\Delta$ Quantization:\par
	A sequence of auxiliary variables $\{u_n\}_{n\in\ZZ}$ is introduced for this scheme. $\{q_n\}_{n\in\ZZ}$ is defined recursively as
\[
\begin{split}
	q_n&=Q_0(u_{n-1}+f(nT)),\\
	u_n&=u_{n-1}+f(nT)-q_n.
\end{split}
\]
\end{itemize}

	$\Sigma\Delta$ quantization was introduced in 1963, \cite{HI_YY_1963}, and is still widely used, due to its advantages over PCM. Specifically, $\Sigma\Delta$ quantization is robust against hardware imperfection \cite{ID_VV_2006}, a decisive weakness for PCM. For $\Sigma\Delta$ quantization, and the more general noise shaping schemes to be explained in Section \ref{NS_scheme}, the boundedness of $\{u_n\}_{n\in\ZZ}$ turns out to be essential, as most analyses on quantization problems rely on it for error estimation. Schemes with bounded auxiliary variables are said to be \textit{stable}.\par
	
	Despite its merits over PCM, $\Sigma\Delta$ quantization merely produces linear error decay with respect to bits used as opposed to exponential error decay produced by its counterpart PCM. Thus, it is desirable to generalize $\Sigma\Delta$ quantization for higher order error decay.
	
	Given $r\in\NN$, one can consider an \textit{$r$-th order $\Sigma\Delta$ quantization scheme} as investigated by Daubechies and DeVore:
	
\begin{theorem}[Higher Order $\Sigma\Delta$ Quantization, \cite{ID_RD_2003}]
	Consider the following stable quantization scheme
\[
	f(nT)-q_n=(\Delta^ru)_n:=\sum_{l=0}^r(-1)^l\binom{r}{l}u_{n-l},
\]
	where $\{q_n\}$ and $\{u_n\}$ are the quantized samples and auxiliary variables respectively, and $(\Delta h)_n:=h_n-h_{n-1}$ for any sequence $h$. Then, for all $t\in\RR$,
\[
	|f(t)-T\sum_{n\in\ZZ}q_n g(t-nT)|\leq T^r\|u\|_\infty\|\frac{d^r g}{dt^r}\|_1.
\]
\end{theorem}

\begin{remark}
\label{back_diff_def}
	The backward difference operator $\Delta$ defined above has a counterpart in finite dimensional spaces. In particular, an $m$-dimensional backward difference matrix $\Delta\in\NN^{m\times m}$ is a lower triangular matrix with unit diagonal entries and $-1$ as sub-diagonal ones. All other entries are identically $0$.
\end{remark}

\subsection{Motivation}

	Higher order $\Sigma\Delta$ quantization has been known for a long time \cite{WC_PW_RG_1989, PF_AG_RA_1990}, and the $r$-th order $\Sigma\Delta$ quantization improves the error decay rate from linear to polynomial degree $r$ while preserving the advantages of a first order $\Sigma\Delta$ quantization scheme. 
	
	However, even the $r$-th order polynomial decay pales in the presence of exponential decay, and thus a natural question arises: is it possible to generalize $\Sigma\Delta$ quantization scheme further so that the reconstruction error decay matches the exponential decay of PCM? Two solutions have been proposed for this question. One is to create new quantization schemes, known as noise shaping quantization schemes. A brief summary of its development will be provided in Section \ref{NS_scheme}.
	
	The other one is to drastically enhance data storage efficiency while maintaining the same level of reconstruction accuracy, and signal decimation belongs in this category. The process is as follows: given an r-th order $\Sigma\Delta$ quantization scheme, there exists $\{q_n^T\}, \{u_n\}$ such that\\
\[
	f(nT)-q_n^T=(\Delta^r u)_n,
\]
	where $\|u\|_\infty<\infty$, and $u_n=0$ for $n\leq0$. Then, consider
\[
	\tilde{q}_n^{T_0}:=(S_\rho^r q^T)_{(2\rho+1)n},
\]
	a down-sampled sequence of $S_\rho^r q^T$, where $(S_\rho h)_n :=\frac{1}{2\rho+1}\sum_{m=-\rho}^{\rho}h_{n+m}$. Signal decimation is the process with which one converts the quantized samples $\{q_n^T\}$ to $\{\tilde{q}_n^{T_0}\}$. See Figure \ref{Diagram} for an illustration.

	Decimation has been known in the engineering community \cite{JC_1986}, and it was observed that decimation results in exponential error decay with respect to the bit-rate, even though the observation remained a conjecture until 2015 \cite{ID_RS_2015}, when Daubechies and Saab proved the following theorem:

\begin{theorem}[Signal Decimation for Bandlimited Functions, \cite{ID_RS_2015}]\label{ID_RS_Deci}
	Given $f\in PW_{1/2}$, $T<1$, and $T_0=(2\rho+1) T< 1$, there exists a function $\tilde{g}$ such that\\
\[
	f(t)=T_0\sum[S_\rho^r f^{(T)}]_{(2\rho+1)n}\tilde{g}(t-nT_0),
\]
\begin{equation}
\label{err_dec_DS}
	|f(t)-T_0\sum \tilde{q}_n^{T_0}\tilde{g}(t-nT_0)|\leq C_{\Sigma\Delta}C^{r}\big(\frac{T}{T_0}\big)^r =:\mathcal{D}.
\end{equation}

	Moreover, the bits needed for each Nyquist interval is
\begin{equation}
\label{resource_DS}
	\frac{1}{T_0}\log_{2}((2\rho+1)^r+1)\leq\frac{1}{T_0}\log_{2}\bigg(2\bigg(\frac{T_0}{T}\bigg)^r\bigg)=:\mathcal{R}.
\end{equation}

	Consequently,\\
\[
	\mathcal{D}(\mathcal{R})=2C_{\Sigma\Delta}C^r 2^{-T_0\mathcal{R}}
\]

\end{theorem}

	From \eqref{err_dec_DS} and \eqref{resource_DS}, we can see that the reconstruction error after decimation still decays polynomially with respect to the sampling frequency. As for the data storage, the bits needed changes from $O(T^{-1})$ to $O(\log(1/T))$. Thus, the reconstruction error decays exponentially with respect to the bits used.

\begin{figure}
\centering
\includegraphics[scale=1]{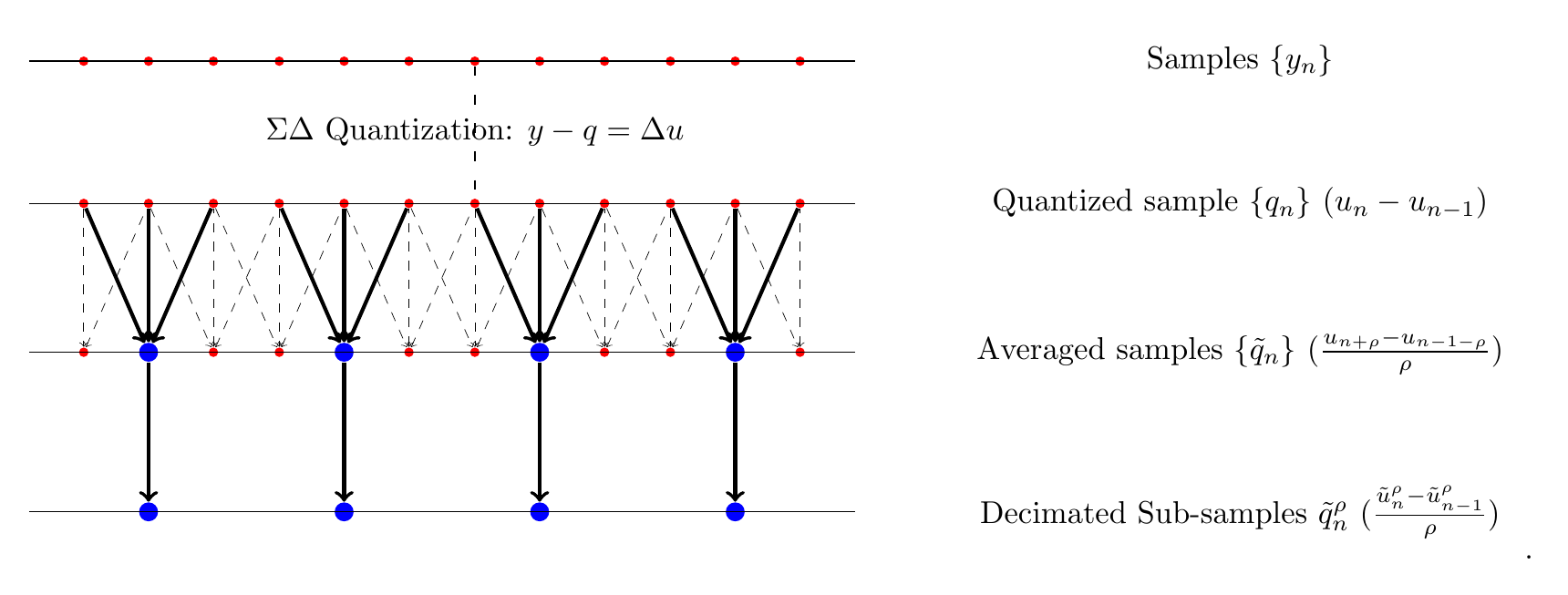}
\caption{Illustration of the first order decimation scheme for A/D conversion. After obtaining the quantized samples $\{q_n\}_n$ in the first step, decimation takes the average of quantized samples within disjoint blocks in the second step. The outputs are used as the decimated sub-samples $\{\tilde{q}_n^\rho\}$ in the third step. The effect on the reconstruction (replacing $q_n$ with $y_n-q_n$) is illustrated in parentheses.}
\label{Diagram}
\end{figure}
	
	Our motivation is the result from Theorem \ref{ID_RS_Deci}. As this theorem is only applicable for A/D conversion, we are interested in extending decimation to finite frames. In particular, we would like to obtain polynomial error decay rate with respect to the oversampling ratio $\rho$ while compressing the data to the order of $O(\log(1/\rho))$. 
	In \cite{KL_2018}, the author made an extension of decimation to signals in finite dimensional spaces. Such signals are sampled by finite frames, and a brief introduction on finite frames is given in Section \ref{prelim}. Using the alternative decimation operator introduced in the same paper, it is proven that up to the second order sigma-delta quantization, similar results to Theorem \ref{ID_RS_Deci} can be achieved. The precise statement will be given in Theorem \ref{KC_Alt}.

		
\subsection{Results and Outline}\label{outline}

	In this paper, we build on our past result in Theorem \ref{KC_Alt} to formulate and prove Theorem \ref{main_uni}. Specifically, Theorem \ref{KC_Alt} is an extension of signal decimation to finite frames up to the second order $\Sigma\Delta$ quantization in \cite{KL_2018}, and Theorem \ref{main_uni} further generalizes Theorem \ref{KC_Alt} to arbitrary orders. We shall show that for any stable $r$-th order $\Sigma\Delta$ quantization, the adapted decimation to be introduced in Section \ref{main_contri} coupled with the quantization scheme yields reconstruction error decay rate of polynomial degree $r$ with respect to the oversampling ratio. Furthermore, thanks to the efficient data storage enabled by adapted decimation, the error decay rate is exponential with respect to the total number of bits used.
	
	To provide necessary background information, we include preliminaries for signal quantization theory on finite frames in Section \ref{prelim}. We first define $\Sigma\Delta$ quantization on finite frames in Section \ref{SD_quant_frame}. Then, we give a formal definition of noise shaping schemes, which is more general than $\Sigma\Delta$ quantization, in Section \ref{NS_scheme}. We define the notion of unitarily generated frames in Section \ref{UGF_def}, which is the class of frames we consider in this paper. Section \ref{prior_work} is devoted to perspective and prior works, and our notation is defined in Section \ref{notation}.
	
	In Section \ref{main_contri}, we first define alternative decimation and state the result of it in Theorem \ref{KC_Alt}, which is proven in \cite{KL_2018}. Then, we define adapted decimation and state our main results in Theorem \ref{main_uni}. We prove Theorem \ref{main_uni} in Section \ref{proof}, and the strategy of our proof is explained in Section \ref{roadmap}.

\section{Preliminaries on Finite Frame Quantization}\label{prelim}
	

	Signal quantization theory on finite frames is well motivated from the need to deal with data corruption or erasure \cite{VG_JK_MV_1999, VG_JK_JK_2001}. The authors considered the PCM quantization scheme described above and modeled the quantization error as random noise. In \cite{JB_AP_OY_2006}, deterministic analysis on $\Sigma\Delta$ quantization for finite frames showed that a linear error decay rate is obtained with respect to the oversampling ratio. Moreover, if the frame satisfies certain smoothness conditions, the decay rate can be super-linear for first order $\Sigma\Delta$ quantization. Noise shaping schemes for finite frames have also been investigated, some of which yield exponential error decay rate \cite{CSG_2015,EC_CSG, CSG_2016}. In this section, we shall provide necessary information on quantization for finite frames before stating our results in Section \ref{main_contri}.

\subsection{$\Sigma\Delta$ Quantization on Finite Frames}\label{SD_quant_frame}
\subsubsection{Overview on Frame Theory}
Fix a separable Hilbert space $\Hcal$ along with a set of vectors $T=\{e_j\}_{j\in\ZZ}\subset\Hcal$. The collection of vectors $T$ forms a frame for $\Hcal$ if there exist $A,B>0$ such that for any $v\in\Hcal$, the following inequality holds:
\[
	A\|v\|_{\Hcal}^2\leq \sum_{j\in\ZZ}|{<}v,e_j{>}|^2\leq B\|v\|_{\Hcal}^2.
\]

	The concept of frames is a generalization of orthonormal bases in a vector space. Different from bases, frames are usually over-complete: the vectors form a linearly dependent spanning set. Over-completeness of frames is particularly useful for noise reduction, and consequently frames are more robust against data corruption than orthonormal bases.\par
	
	Let us restrict ourselves to the case when $\Hcal=\CC^k$ is a finite dimensional Euclidean space, and the frame consists of a finite number of vectors. Given a finite frame $T=\{e_j\}_{j=1}^m$, the linear operator $E:\CC^k\to\CC^m$ satisfying $Ev=\{{<}v,e_j{>}\}_{j=1}^m$ is called the \textit{analysis operator}. Its adjoint operator $E^\ast:\CC^m\to\CC^k$ satisfies $E^\ast c=\sum_{j=1}^mc_je_j$ and is called the \textit{synthesis operator}. The \textit{frame operator} $\Scal$ is defined by $\Scal=E^\ast E: \CC^k\to\CC^k$.

\begin{remark}\label{frame_bound_remark}
	Note that since $\Scal$ is Hermitian, 
\[
\|\Scal\|_2=\max_{v: \|v\|_2=1}|v^T\Scal v|=\max_{v:\|v\|_2=1}\sum_{j=1}^m|{<}v,e_j{>}|^2\leq B. 
\]
	Similarly, $\|\Scal^{-1}\|_2\leq A^{-1}$. In particular, the 2-norm of $\Scal$ is directly tied to the lower frame bound of $T$.
\end{remark}

	Under this framework, one considers the quantized samples $q$ of $y=Ex$ and reconstructs $\tilde{x}=\Scal^{-1}E^\ast q$, where $\Scal=E^\ast E$. 
	
\subsubsection{$\Sigma\Delta$ Quantization and Mid-Rise Uniform Quantizers}
	The frame-theoretic $r$-th order greedy $\Sigma\Delta$ quantization is defined as follows: given a finite alphabet $\mathscr{A}\subset\mathbb{C}$ and $r\in\NN$, we calculate $q, u\in\CC^m$ as follows:
\begin{equation}
\label{recur}
		y-q=\Delta^r u,
\end{equation}
	where $\Delta\in\mathbb{Z}^{m\times m}$ is the backward difference matrix. The quantization scheme is said to be \textit{stable} if $\|u\|_\infty$ is uniformly bounded for all $m$. For the rest of the paper, we shall assume that such stable schemes exist.
	
	In practice, the quantization alphabet $\mathscr{A}$ is often chosen to be $\mathscr{A}_0$ which is uniformly spaced and symmetric around the origin: given $\delta>0$, we define a mid-rise uniform quantizer $\mathscr{A}_0$ of length $2L$ to be $\mathscr{A}_0=\{(2j+1)\delta/2:\, -L\leq j\leq L-1\}$.

	For complex Euclidean spaces, we define $\mathscr{A}=\mathscr{A}_0+\imath\mathscr{A}_0$. In both cases, $\mathscr{A}$ is called a mid-rise uniform quantizer. Throughout this paper we shall always be using $\mathscr{A}$ as our quantization alphabet.\par

	
\subsection{Noise Shaping Schemes and the Choice of Dual Frames}\label{NS_scheme}

	$\Sigma\Delta$ quantization is a subclass of the more general noise shaping quantization, where the quantization scheme is designed such that the reconstruction error is easily separated from the true signal in the frequency domain. For instance, it is pointed out in \cite{CSG_2016} that the reconstruction error of $\Sigma\Delta$ quantization for bandlimited functions is concentrated in high frequency ranges. Since audio signals have finite bandwidth, it is then possible to separate the signal from the error using low-pass filters. 
	
	Noise shaping quantization has been well established for A/D conversion since the mid 20th century \cite{ST_RH_1978}, and in terms of finite frames, noise shaping schemes generalize the $\Sigma\Delta$ scheme in the following way:
\[
	y-q=Hu,
\]
	where $y, q$, and $u$ are the samples, quantized samples, and the auxiliary variable, respectively, while the transfer matrix $H$ is lower-triangular. Now, given an analysis operator $E$, a transfer matrix $H$, and a dual $F$ to $E$, i.e.\ , $FE=I_k$, the reconstruction error in this setting is 
\[
\|x-Fq\|_2=\|F(Ex-q)\|_2=\|FHu\|_2\leq \|FH\|_{\infty,2}\|u\|_\infty, 
\]
	where $\|\cdot\|_{\infty,2}$ is the operator norm between $\ell^\infty$ and $\ell^2$, i.e.,
\[
	\|T\|_{\infty,2}:=\sup_{\|x\|_\infty=1}\|Tx\|_2.
\]

	The choice of the dual frame $F$ plays a role in the reconstruction error. For instance, \cite{JB_ML_AP_OY_2010} proved that $\argmin_{FE=I_k}\|FH\|_2=(H^{-1}E)^\dagger H^{-1}$, where given any matrix $A$, $A^\dagger$ is defined as the canonical dual $(A^\ast A)^{-1}A^\ast$. More generally, one can consider a $V$-dual, namely $(VE)^\dagger V$, provided that $VE$ is still a frame. With this terminology, decimation can be viewed as a special case of $V$-duals, and conversely every $V$-dual can be associated with corresponding post-processing on the quantized sample $q$.
	
\subsection{Unitarily Generated Frames}\label{UGF_def}
	
	In this paper, we are interested in a specific class of frames called the unitarily generated frames (UGF). A unitarily generated frame $T_u$ is generated by a cyclic group: given a unit base vector $\phi_0\in\CC^k$ and a Hermitian matrix $\Omega\in\CC^{k\times k}$, the frame elements of $T_u$ are defined as
\[
	\phi_j^{(m)}=U_{j/m}\phi_0,\quad   U_t:=e^{2\pi\imath\Omega t}.
\]
	The analysis operator $\Phi$ of $T_u$ has $\{\phi_j^\ast\}_j$ as its rows.

	As symmetry occurs naturally in many applications, it is not surprising that unitarily generated frames receive serious attention, and their applications in signal processing abound, \cite{GF_1991, HB_YE_2003, EC_CSG, CSG_2016}. \par
	
	One particular application comes from dynamical sampling, which records the spatiotemporal samples of a signal in interest. Mathematically speaking, one tries to recover a signal $f$ on a domain $D$ from the samples $\{f(X),f_{t_1}(X),\dots, f_{t_N}(X)\}$ where $X\subset D$, and $f_{t_j}=A^{t_j}f$ denotes the evolved signal. Equivalently, one recovers $f$ from $\{{<}A^{t_j}f,e_i{>}\}_{i,j}=\{{<}f,(A^{t_j})^{\ast}e_i{>}\}_{i,j}$, which aligns with the frame reconstruction problems, \cite{AA_JD_IK_2013, AA_JD_IK_2015}. In particular, Lu and Vetterli \cite{YL_MV_2009, YL_MV_2009_2} investigated the reconstruction from spatiotemporal samples for a diffusion process. They noted that one can compensate under-sampled spatial information with sufficiently over-sampled temporal data. Unitarily generated frames represent the cases when the evolution process is unitary and the spatial information is one-dimensional. 
	
	It should be noted that unitarily generated frames are group frames with the generator $G=U_{1/m}$ provided that $U_1=G^m=I_k$, while harmonic frames are tight unitarily generated frames. Here, a frame $T=\{e_j\}_j\subset\Hcal$ is tight if for all $v\in \Hcal$, there exists a constant $A>0$ such that $\sum_{j}|{<}v,e_j{>}|^2=A\|v\|^2$.
	
	 A common class of harmonic frames is the exponential frame with generator $\Omega$ as a diagonal matrix with integer entries and the base vector $\phi_0=(1,\dots,1)^t/\sqrt{k}$.\par


	
\subsection{Perspective and Prior Works}\label{prior_work}

\subsubsection{Quantization for Bandlimited Functions}
	Despite its simple form and robustness, $\Sigma\Delta$ quantization only results in linear error decay with respect to the sampling period $T$ as $T\to0$. It was later proven in \cite{ID_RD_2003} that a generalization of $\Sigma\Delta$ quantization, namely the r-th order $\Sigma\Delta$ quantization, exists for any arbitrary $r\in\NN$, and for such schemes the error decay is of polynomial order $r$. Leveraging the different constants for this family of quantization schemes, sub-exponential decay can also be achieved. A different family of quantization schemes was shown \cite{CSG_2003} to have exponential error decay with small exponent ($c\approx0.07$.) In \cite{PD_CG_FK_2011}, the exponent was improved to $c\approx 0.102$.\par

\subsubsection{Finite Frames}
	$\Sigma\Delta$ quantization can also be applied to finite frames. It is proven \cite{JB_AP_OY_2006} that for any family of frames with bounded frame variation, the reconstruction error decays linearly with respect to the oversampling ratio $m/k$, where the frame is an $m\times k$ matrix. With different choices of dual frames, \cite{JB_ML_AP_OY_2010} proposed that the so-called Sobolev dual achieves minimum induced matrix 2-norm for reconstructions. The limit of $\Sigma\Delta$ quantization for arbitrary frames is detailed in \cite{BB_VP_SA_2007}. Using smooth frame-path with vanishing derivatives at the endpoint yields polynomial error decay rate for higher order $\Sigma\Delta$ quantization. By carefully matching between the dual frame and the quantization scheme, \cite{CSG_2016} proved that using $\beta$-dual for random frames will result in exponential decay with near-optimal exponent and high probability.\par
	
\subsubsection{Decimation}
	In \cite{JC_1986}, using the assumption that the noise in $\Sigma\Delta$ quantization is random along with numerical experiments, it was asserted that decimation greatly reduces the number of bits needed while maintaining the reconstruction accuracy. In \cite{ID_RS_2015}, a rigorous proof was given to show that such an assertion is indeed valid, and the reduction of bits used turns the linear decay into exponential decay with respect to the bit-rate.
	
	Adapting decimation to finite frames is by no means a new idea. Iwen and Saab \cite{MI_RS_2013} used probabilistic arguments and the property of efficient storage to construct random quantization schemes with exponential error decay rate with respect to the bit usage. In \cite{TH_RS_2018}, similar ideas are used on $\Sigma\Delta$. Moreover, the connection between decimation and distributed noise shaping can be seen in it.
	
	\cite{MI_RS_2013, TH_RS_2018} both use probabilistic arguments that only ensure success with some probability instead of deterministic guarantee. For the explicit and deterministic adaptation to finite dimensional signals, the author proved in \cite{KL_2018} that there exists a similar operator called the alternative decimation operator that behaves similarly to the decimation for bandlimited signals. In particular, for the first and second order of $\Sigma\Delta$ quantization, it is possible to achieve exponential reconstruction error decay with respect to the bit-rate as well. However, similar to the caveat of decimation, it merely improves the storage efficiency while maintaining the same level reconstruction error. Thus, the error rate with respect to the oversampling ratio remains the same (quadratic for the second order,) which is still inferior to other noise shaping schemes.
	
\subsubsection{Beta Dual of Distributed Noise Shaping}\label{beta_dual}
	Chou and G{\" u}nturk \cite{CSG_2016, EC_CSG} proposed a distributed noise shaping quantization scheme with beta duals as an example. The definition of a beta dual is as follows:
\begin{definition}[Beta Dual]
	Let $E\in\RR^{m\times k}$ be an analysis operator and $k\mid m$. Recall that $F_V\in\RR^{k\times m}$ is a V-dual of $E$ if\\
\[
	F_V=(VE)^{\dagger}V,
\]
	where $V\in\RR^{p\times m}$ such that $VE$ is still a frame.
\end{definition}

	Given $\beta>1$, the $\beta$-dual $F_V=(VE)^{\dagger}V$ has  $V=V_{\beta,m}$, a $k$-by-$m$ block matrix such that each block is $v=[\beta^{-1},\beta^{-2},\dots,\beta^{-m/k}]\in\RR^{1\times m/k}$.
	
	In this case, the transfer matrix $H$ is an $m$-by-$m$ block matrix where each block $h$ is an $m/k$-by-$m/k$ matrix with unit diagonal entries and $-\beta$ as sub-diagonal entries. Under this setting, it is proven that the reconstruction error decays exponentially.
	
	One may notice the similarity between the beta dual and decimation. Indeed, if one chooses $\beta=1$ and normalizes $V$ by $\frac{k}{m}$, the same result as decimation can be obtained, achieving linear error decay with respect to the oversampling ratio and exponential decay with respect to the bit usage. Nonetheless, its generalization to higher order error decay with respect to the oversampling ratio is lacking, whereas the adapted decimation we propose can be extended to arbitrary polynomial degrees.

\subsection{Notation}\label{notation}
	The following notation is used in this paper:
\begin{itemize}
\item	$x\in\CC^k$: the signal of interest.
\item	$\Omega\in\CC^{k\times k}$: a Hermitian matrix with eigenvalues $\{\lambda_j\}_{j=1}^k\subset\RR$ and corresponding orthonormal eigenvectors $\{v_j\}_{j=1}^k$.
\item	$\Phi\in\CC^{m\times k}$: the analysis operator of the unitarily generated frame (UGF) with the generator $\Omega$ and the base vector $\phi_0\in\CC^k$.
\item	$U_t\in\CC^{k\times k}$: the unitary matrix defined as $U_t=e^{2\pi\imath\Omega t}$ for any $t\in\RR$.
\item	$B=B_\Phi\in\CC^{k\times k}$: a unitary matrix that simultaneously diagonalizes $U_t$ and $\Omega$. In particular, $\Omega=B\Lambda B^\ast$ and $U_t=B e^{2\pi\imath\Lambda t}B^\ast$, where $\Lambda=diag(\lambda_1,\dots, \lambda_k)$.
\item $y=\Phi x\in\CC^m$: the sample.
\item	$q\in\CC^m$: the quantized sample obtained from the greedy $\Sigma\Delta$ quantization defined in \eqref{recur}.
\item	$u\in\CC^m$: the auxiliary variable of $\Sigma\Delta$ quantization.
\item	$\rho\in\NN$: the block size of the decimation.
\item	$\eta= m/\rho\in\NN$: the dimension of compressed data.
\item $\mathscr{A}=\mathscr{A}_0+\imath\mathscr{A}_0\subset\CC$: the quantization alphabet. $\mathscr{A}$ is said to have length $2L$ with gap $\delta$ if $\mathscr{A}_0=\{(2j+1)\delta/2:\, -L\leq j\leq L-1\}$ for some $\delta>0$.
\item $F\in\CC^{k\times m}$: a dual to the analysis operator $\Phi$, i.e.\ $F\Phi=I_k$.
\item	$\mathscr{E}$: the reconstruction error $\mathscr{E}=\|x-Fq\|_2$.
\item	$\mathscr{R}$: total number of bits used to record the quantized sample.
\item	$\|\cdot\|_{p,q}$: the $p$-to-$q$ norm. For any matrix $M$, $\|M\|_{p,q}:=\sup_{v: \|v\|_p=1}\|Mv\|_q$. For simplicity, we denote $\|\cdot\|_2:=\|\cdot\|_{2\to 2}$ for matrices.
\item	$\delta:\ZZ\to\{0,1\}$: the Kronecker delta. $\delta(k)=1$ if $k=0$, and $0$ otherwise. With some abuse of notation, we may also view $\delta$ as a function on the cyclic group $\ZZ/\ell\ZZ$ for any $\ell\in\NN$.
\end{itemize}

\section{Contributions}\label{main_contri}

	In Theorem \ref{ID_RS_Deci}, we see that signal decimation coupled with the $r$-th $\Sigma\Delta$ quantization scheme yields polynomial error decay rate of degree $r$ with respect to the oversampling ratio. Moreover, it yields exponential error decay rate the bit-rate. The question we seek to address is whether it is possible to translate decimation from A/D conversion to finite frame quantization. This adaptation proves to be non-trivial, as the $r$-th order $\Sigma\Delta$ quantization does not yield much more than linear error decay rate for finite frames in general as opposed to polynomial degree $r$, \cite{JB_AP_OY_2006, KL_2018}. 
	
	With the introduction of \textit{alternative decimation}, the author was able to adapt signal decimation to finite frames up to the second order $\Sigma\Delta$ quantization \cite{KL_2018}, yielding quadratic error decay rate with respect to the oversampling ratio. This paper further generalizes the concept of decimation and extends the decimation on finite frames to arbitrary polynomial degrees. 
	
	For the sake of completeness, we briefly formulate alternative decimation and the corresponding results below:

\subsection{Past Result: Alternative Decimation}\label{past_result}
\begin{definition}[Alternative Decimation]
	Given fixed $m,\rho\in\NN$, the $(r,m,\rho)$-alternative decimation operator is defined to be $D_\rho S_\rho^r$, where
\begin{itemize}
\item	$ S_\rho=S_\rho^+-S_\rho^-\in\RR^{m\times m}$ is the integration operator satisfying
\[
\begin{split}
		&(S_\rho^+)_{l,j}=\left\{\begin{array}{ll}
	\frac{1}{\rho}&\text{if}\quad l\geq\rho,\, l-(\rho-1)\leq j\leq l\\
	0&\text{otherwise},
	\end{array}\right.\\
	&(S_\rho^-)_{l,j}=\left\{\begin{array}{ll}
	\frac{1}{\rho}&\text{if}\quad l\leq\rho-1,\, l+1\leq j\leq m-\rho+l \\
	0&\text{otherwise.}
	\end{array}\right.
\end{split}
\]

	Here, the cyclic convention is adopted: For any $s\in\ZZ$, $s\equiv s+m$.
\item	$D_\rho\in\NN^{p\times m}$ is the sub-sampling operator satisfying
\[
	(D_\rho)_{l,j}=\left\{\begin{array}{ll}
		1&\text{if}\quad j=\rho\cdot l\\
		0&\text{otherwise},
	\end{array}\right.
\]
	and $\eta=\lfloor m/\rho\rfloor$.
\end{itemize}
\end{definition}

\begin{theorem}[Alternative Decimation for Finite Frames, \cite{KL_2018}]
\label{KC_Alt}
	Given $\Omega$, $\phi_0$, $\{\lambda_j\}_j$, $\{v_j\}_j$, and $\Phi=\Phi_{m,k}$ as the generator, base vector, eigenvalues, eigenvectors, and the corresponding UGF, respectively, and \textbf{r=1,2}. Suppose
\begin{itemize}
\item	$\{\lambda_j\}_{j=1}^k\subset[-\eta/2,\eta/2]\cap\ZZ\backslash\{0\}$,
\item $\min_s |{<}\phi_0,v_s{>}|^2>0$, and
\item	$\rho\mid m$,
\end{itemize}
	
	then the dual frame $F=(D_\rho S_\rho^r \Phi_{m,k})^\dagger D_\rho S_\rho^r$ combined with the $r$-th order $\Sigma\Delta$ quantization has polynomial reconstruction error decay rate of degree $r$ with respect to the oversampling ratio $\rho$:
\[
	\mathscr{E}_{m,\rho,r}\leq C\|u\|_\infty\frac{1}{\rho^r}.
\]
	
	Moreover, the total bits used to record the quantized samples are $\mathscr{R}=O(\log(m))$ bits, where the constant depends on $r$. Suppose $m/\rho=\eta$ is fixed as $m\to\infty$, then as a function of bits used at each entry, $\mathscr{E}_{m,\rho}$ satisfies
\[
	\mathscr{E}(\mathscr{R})\leq C\|u\|_\infty 2^{-\frac{1}{2\eta}\mathscr{R}}.
\]
	
	The constant $C$ is independent of the oversampling ratio $\rho$.

\end{theorem}

\subsection{Main Result: Adapted Decimation}	
	We have seen in Theorem \ref{KC_Alt} that alternative decimation is only useful up to the second order. Thus, we aim to extend our results to arbitrary orders, and the solution we present here is called \textit{adapted decimation}.

\begin{definition}[Adapted Decimation]
  Given $r,m,\rho\in\NN$, the $(r,m,\rho)$-adapted decimation operator is defined to be
  \[
  A_r=\frac{1}{\rho^r}D_\rho \bar{\Delta}^r_\rho\Delta^{-r},
  \]
  where $\Delta\in\NN^{m\times m}$ is the usual backward difference matrix, $\bar{\Delta}_\rho\in\RR^{m\times m}$ satisfies $(\bar{\Delta}_\rho)_{l,s}=\frac{1}{\rho}(\delta(l-s)-\delta(l+\rho-s)+\delta(s-m)\delta(l-\rho))$, and $D_\rho\in\NN^{m/\rho\times m}$ has $(D_\rho)_{l,s}=\delta(s-l\rho)$.

\end{definition}

\begin{figure}
\centering
\includegraphics[scale=0.9]{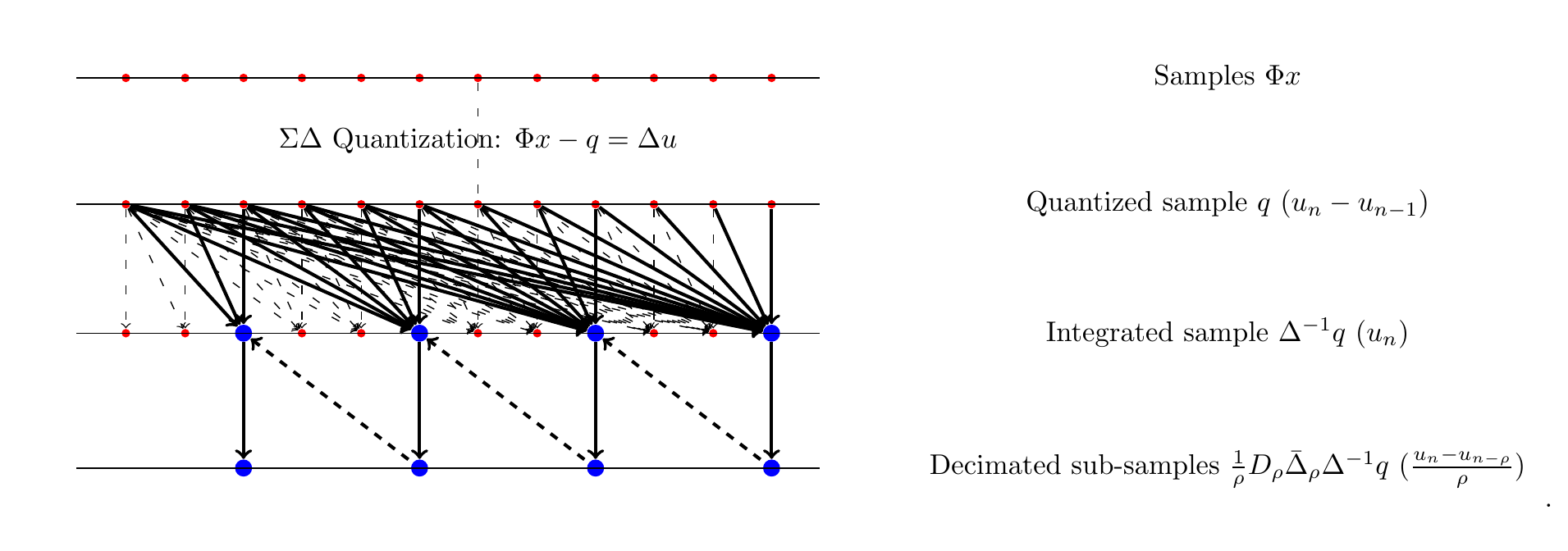}
\caption{Illustration of the first order adapted (alternative) decimation scheme for finite frames. After obtaining the quantized samples $\{q_n\}_n$ in the first step, one starts by integrating quantized samples in the second step. Finite difference of step size $\rho$ followed by sub-sampling are then taken in the third step. The effect on the reconstruction (replacing $q_n$ with $y_n-q_n$) is illustrated in parentheses. Note that both the recursivity and the boundary effect (see bottom left) can be seen in this diagram.}
\label{Diagram_AdaDec}
\end{figure}

\begin{remark}[Comparison between Alternative and Adapted Decimation]\label{Alt_vs_Ada}
	While coinciding for $r=1$, $A_r$ is different from $D_\rho S_\rho^r$ in the following way: $S_\rho=\frac{1}{\rho}\bar{\Delta}_\rho\Delta^{-1}$, and thus
\[
	D_\rho S_\rho^2=\frac{1}{\rho^2}D_\rho(\bar{\Delta}_\rho\Delta^{-1})^2\neq \frac{1}{\rho^2}D_\rho\bar{\Delta}_\rho^2\Delta^{-2}=A_2.
\]
	The non-commutativity between $\bar{\Delta}_\rho$ and $\Delta^{-1}$ limits the success of the alternative decimation, see Proposition A.1 in \cite{KL_2018}. Adapted decimation essentially factorizes the alternative decimation and re-arranges the terms. In doing so, the reconstruction error rate can now be of polynomial degree $r$. However, it also complicates the effect of decimation on finite frames, as will be seen in Section \ref{Ar_form}. For the illustration, see Figure \ref{Diagram_AdaDec}.
\end{remark}

	It will be shown that, for unitarily generated frames $\Phi\in\CC^{m\times k}$ satisfying conditions specified in Theorem \ref{main_uni} and \textit{any} $r\in\NN$, an $r$-th order $\Sigma\Delta$ quantization coupled with the corresponding adapted decimation has $r$-th order polynomial reconstruction error decay rate with respect to the ratio $\rho$. As for the data storage, decimation allows for highly efficient storage, making the error decay exponentially with respect to the bit usage.

\begin{theorem}
\label{main_uni}
	Given $\Omega$, $\phi_0$, $\{\lambda_j\}_j$, $\{v_j\}_j$, and $\Phi=\Phi_{m,k}$ as the generator, base vector, eigenvalues, eigenvectors, and the corresponding UGF, respectively, and $r\in\NN$ fixed. Suppose
\begin{itemize}
\item	$\rho\mid m$,
\item	$\eta=m/\rho\geq 3rk$,
\item	$\{\lambda_j\}_{j=1}^k\subset[-\eta/2,\eta/2]\cap\ZZ\backslash\{0\}$, and
\item $C_{\phi_0}=\min_s |{<}\phi_0,v_s{>}|^2>0$,
\end{itemize}
	then the following statements are true:
\begin{itemize}
\item[(a)] {\bf Recursivity:} For all $s\in\{1,\dots, \eta\}$, there exists $\{c_j^s\}_{j=1}^{s\rho}$ such that $(A_r q)_s=\sum_{j=1}^{s\rho}c_j^s q_j$.
\item[(b)] {\bf Signal reconstruction:} $A_r \Phi_{m,k}$ is a frame. 
\item[(c)] {\bf Error estimate:} Given the dual frame $F=(A_r \Phi_{m,k})^\dagger A_r$, where for any $M$, $M^\dagger=(M^\ast M)^{-1}M^\ast$ is defined to be the pseudo-inverse of $M$. Then the reconstruction error $\mathscr{E}_{m,\rho}=\|x-Fq\|_2$ satisfies
\begin{equation}
\label{error_est_uni}
	\mathscr{E}_{m,\rho}\leq\bigg(\frac{4}{k\eta C_{\phi_0}}(\pi^2\eta)^r\bigg)\|u\|_\infty\frac{1}{\rho^r}.
\end{equation}
\item[(d)] {\bf Efficient data storage:} Suppose the length of the quantization alphabet is $2L$, then the total bits used to record the quantized samples $A_r q$ are $\mathscr{R}=2\eta r\log(2m)+2\eta\log(2L)$ bits. Furthermore, as a function of bits used at each entry, $\mathscr{E}_{m,\rho}$ satisfies
\[
	\mathscr{E}(\mathscr{R})\leq C_{k,\eta,\phi_0,L}\|u\|_\infty 2^{-\frac{1}{2\eta}\mathscr{R}},
\]
	where $C_{k,\eta,\phi_0,L}=\frac{8L}{k\eta C_{\phi_0}}(2\pi^2)^r$, independent of $\rho$.
\end{itemize}
\end{theorem}


\section{Proof of Main Results}\label{proof}

\subsection{Roadmap of the Proof}\label{roadmap}
        
        In this subsection, we shall identify the key components regarding the proof of Theorem \ref{main_uni}. Then, we will provide estimates for those components in Sections 4.2-4.5 before finishing the proof in Section \ref{main_proof}. 
        
        To estimate the reconstruction error $\mathscr{E}_{m,\rho}=\|x-(A_r\Phi_{m,k})^\dagger A_r q\|_2$ in \eqref{error_est_uni}, we re-write the form of $A_r$, making the estimate simpler. In particular, we claim that $\bar{\Delta}_\rho$ scales down to the usual backward-difference matrix under the under-sampling matrix $D_\rho$:

\begin{lemma}
\label{twist_scale}
	Given $m,\rho\in\NN$ with $\eta=m/\rho\in\NN$,	
\[
D_\rho\bar{\Delta}_\rho=\Delta^{(\eta)}D_\rho,
\]
\end{lemma}
	where $\Delta^{(\eta)}$ is the $\eta$-dimensional backward difference matrix.
\begin{proof}

	Note that, for $s\neq m$,
\[
\begin{split}
	(D_\rho\bar{\Delta}_\rho)_{l,s}&=(\bar{\Delta}_\rho)_{l\rho,s}\\
			&=\delta(s-l\rho)-\delta(s+\rho-l\rho)=\delta(s-l\rho)-\delta(s-(l-1)\rho)\\
			&=(\Delta D_\rho)_{l,s}.
\end{split}
\]
	For $s=m$, $(D_\rho\bar{\Delta}_\rho)_{l,m}=\delta(m-l\rho)=(\Delta D_\rho)_{l,m}$.
\end{proof}

	Then, the reconstruction error $\mathscr{E}_{m,\rho}$ satisfies

\begin{equation}
\label{E_est_RP}
\begin{split}
	\mathscr{E}_{m,\rho}&=\|x-(A_r\Phi_{m,k})^\dagger A_r q\|_2\\
				&=\|(A_r\Phi_{m,k})^\dagger A_r(\Phi_{m,k} x-q)\|_2\\
				&=\|(A_r\Phi_{m,k})^\dagger \frac{1}{\rho^r}D_\rho\bar{\Delta}^r\Delta^{-r}(\Delta^r u)\|_2\\
				&=\frac{1}{\rho^r}\|\big((A_r\Phi_{m,k})^\ast A_r\Phi_{m,k}\big)^{-1}(A_r\Phi_{m,k})^\ast \Delta^r D_\rho u\|_2\\
				&\leq\|\big((A_r\Phi_{m,k})^\ast A_r\Phi_{m,k}\big)^{-1}\|_2\cdot\|(A_r\Phi_{m,k})^\ast\Delta^r\|_{\infty,2}\cdot\|u\|_\infty\frac{1}{\rho^r},
\end{split}
\end{equation}
        where the fourth equality follows from Lemma \ref{twist_scale}. We have seen from Remark \ref{frame_bound_remark} that $\|((A_r\Phi_{m,k})^\ast A_r\Phi_{m,k})^{-1}\|_2$ is the reciprocal of the lower frame bound of $A_r\Phi_{m,k}$. Thus, in order to estimate \eqref{E_est_RP}, we need only to answer two questions: 

\begin{itemize}
\item	Is $A_r\Phi_{m,k}$ a frame? What is the lower frame bound of $A_r\Phi_{m,k}$?
\item	What is $\|(A_r\Phi_{m,k})^\ast\Delta^r\|_{\infty,2}$?
\end{itemize}

	The lower frame bound of $A_r\Phi_{m,k}$ will be calculated in Section \ref{LFBA}, specifically in Proposition \ref{lower_frame_bound}. As for the estimate in the second question, it is given in Proposition \ref{variation_bound} of Section \ref{FVB}.
	
	Aside from the reconstruction error estimate, we also need to calculate the number of bits needed to record the decimated sample $A_r q$. We shall show that $A_r q$ can be efficiently stored in $O(\log \rho)$ instead of $O(\rho)$ bits. The explicit estimate will be done in Proposition \ref{data_storage} of Section \ref{DSE}.

\subsection{Expansion of $A_r\Phi_{m,k}$}\label{Ar_form}

	In \cite{KL_2018}, one has, for any $r\in\NN$, the alternative decimation satisifes 
\[
D_\rho S_\rho^r\Phi_{m,k}=\frac{1}{\rho^r}D_\rho(\bar{\Delta}_\rho\Delta^{-1})^r\Phi_{m,k}=\Phi_{\eta,k}(\tilde{D}\tilde{C})^r
\] 
	where $\tilde{D},\tilde{C}\in\CC^k$ will be defined in Section \ref{EADF}. The form is rather simple thanks to the alternating applications of $\bar{\Delta}_\rho$ and $\Delta^{-1}$. For adapted decimation, we have $A_r=\frac{1}{\rho^r}D_\rho\bar{\Delta}_\rho^r\Delta^{-r}$, and the displaced order of applications creates residual terms other than $\Phi_{\eta,k}(\tilde{D}\tilde{C})^r$. In this section, we observe this phenomenon and examine the effect of the residual terms.

\subsubsection{The Effect of Adapted Decimation on the Frame}\label{EADF}

	We start by introducing the following notation:

\begin{definition}

  Given $l,s\in\NN$, the $l$-by-$s$ constant matrix $1_{l,s}$ has constant $1$ on all entries.

\end{definition}

	The following two lemmas are needed for us to describe $A_r\Phi_{m,k}$ in Proposition \ref{deci_expansion}.

\begin{lemma}
\label{C_lemma}
  Given $\Phi=\Phi_{m,k}\in\CC^{m\times k}$ with base vector $\phi_0$, we have
  \[
  \Delta^{-1}\Phi =(\Phi-1_{m,k}V)\tilde{C},
  \]
  where $\tilde{C}$ and $U_t$ are simultaneously diagonalizable with $B^\ast\tilde{C}B=\tilde{C}_0=diag(\frac{1}{1-e^{2\pi\imath \lambda_s/m}})_{1\leq s\leq m}$ and $V=diag(\phi_0)$.

\end{lemma}

\begin{proof}

	For any $1\leq t\leq m$, the $t$-th row of $\Delta^{-1}\Phi_{m,k}$ can be written as
\[
\begin{split}
	(\Delta^{-1}\Phi_{m,k})_t&=(\sum_{s=1}^t U_{s/m}\phi_0)^\ast\\
						&=(\sum_{s=1}^t BT_{s/m}B^\ast\phi_0)^\ast\\
						&=(B\sum_{s=1}^t T_{s/m} B^\ast\phi_0)^\ast,
\end{split}
\]
	where we note that  $U_t=BT_t B^\ast$ can be diagonalized by the unitary matrix $B=B_\Phi$, and $T_t=e^{2\pi\imath \Lambda t}=diag(\exp(2\pi\imath \lambda_s t))_s$. Now,
\[
\begin{split}
	\sum_{s=1}^t (T_{s/m})_\sigma&=\sum_{s=1}^t e^{2\pi\imath \lambda_\sigma s/m}\\
						&=\frac{e^{2\pi\imath\lambda_\sigma t/m}-1}{e^{2\pi\imath\lambda_\sigma/m}-1}\\
						&=(T_{t/m})_\sigma\frac{1}{e^{2\pi\imath\lambda_\sigma/m}-1}-\frac{1}{e^{2\pi\imath\lambda_\sigma/m}-1}\\
						&=(\tilde{C}_0T_{t/m}-\tilde{C}_0)_\sigma,
\end{split}
\]

	Then,
\[
\begin{split}
		(\Delta^{-1}\Phi_{m,k})_t&=(B\sum_{s=1}^t T_{s/m}B^\ast\phi_0)^\ast\\
							&=(B\tilde{C}_0 B^\ast U_{t/m}\phi_0)^\ast -(B\tilde{C}_0 B^\ast\phi_0)^\ast\\
							&=\phi_t^\ast(B\tilde{C}_0B^\ast)^\ast-\phi_0^\ast (B\tilde{C}_0B^\ast)^\ast\\
							&=(\Phi_{m,k})_t\tilde{C}-\phi_0^\ast\tilde{C}.
\end{split}
\]

	Thus, $\Delta^{-1}\Phi_{m,k}=\Phi_{m,k}\tilde{C}-1_{m,k}V\tilde{C}$.
\end{proof}


\begin{lemma}
\label{D_lemma}
 \[
 \bar{\Delta}_\rho\Phi=\Phi\tilde{D}+\bar{\Delta}_\rho 1_{m,k}V,
 \]
        where $B^\ast\tilde{D}B=diag(1-e^{2\pi\imath\rho n_s/m})_{1\leq s\leq m}$.
\end{lemma}

\begin{proof}

	For any $1\leq t\leq m$,
\[
\begin{split}
	(\bar{\Delta}_\rho \Phi_{m,k})_t&=(U_{t/m}\phi_0 -U_{(t-\rho)/m}\phi_0)^\ast+\delta(t-\rho)\phi_0^\ast\\
			&=(B(I_k-T_{-\rho/m})B^\ast U_{t/m}\phi_0)^\ast+\delta(t-\rho)\phi_0^\ast\\
			&=\phi_{t}^\ast B(I_k-T_{\rho/m})B^\ast+\bar{\Delta}_\rho 1_{m,k}V\\
			&=(\Phi\tilde{D})_t+(\bar{\Delta}_\rho 1_{m,k}V)_t.
\end{split}
\]
\end{proof}


	Combining Lemma \ref{C_lemma} and \ref{D_lemma}, one has the following expansion:

\begin{proposition}
\label{deci_expansion}
  Given $r,m,\rho\in\NN$,
  \begin{equation}
 \label{eq:deci_expansion}
  \rho^r A_r\Phi_{m,k}=D_\rho\bar{\Delta}_\rho^r\Delta^{-r}\Phi_{m,k}=D_\rho\bigg[\Phi_{m,k}\tilde{D}^r\tilde{C}^r+\sum_{j=0}^{r-1}\bar{\Delta}_\rho^{r-j}1_{m,k}V\tilde{D}^{j}\tilde{C}^r-\bar{\Delta}_\rho^r\sum_{j=0}^{r-1}\Delta^{-j}1_{m,k}V\tilde{C}^{r-j}\bigg].
  \end{equation}

\end{proposition}  

\begin{remark}
	Note that $\tilde{D}\tilde{C}=\tilde{C}\tilde{D}$ as they are simultaneously diagonalizable by $B_\Phi$, and thus $\tilde{D}^r\tilde{C}^r=(\tilde{D}\tilde{C})^r$.
\end{remark}

\begin{proof}
	First, we claim that, for $1\leq q\leq r$, $\Delta^{-q}\Phi=\Phi\tilde{C}^q-\sum_{j=0}^{q-1}\Delta^{-j}1_{m,k}V\tilde{C}^{q-j}$.
	
	For $q=1$, $\Delta^{-1}\Phi=\Phi\tilde{C}-1_{m,k}V\tilde{C}$ by Lemma \ref{C_lemma}. For $q>1$,
\[
\begin{split}
	\Delta^{-q}\Phi&=\Delta^{-1}\bigg(\Phi\tilde{C}^{q-1}-\sum_{j=0}^{q-2}\Delta^{-j}1_{m,k}V\tilde{C}^{q-1-j}\bigg)\\
			&=\Phi\tilde{C}^{q}-1_{m,k}V\tilde{C}^q-\sum_{s=1}^{q-1}\Delta^{-s}1_{m,k}V\tilde{C}^{q-s}\\
			&=\Phi\tilde{C}^q-\sum_{j=0}^{q-1}\Delta^{-j}1_{m,k}V\tilde{C}^{q-j}.
\end{split}
\]
	As for the effect of $\bar{\Delta}_\rho$, we claim that $\bar{\Delta}_\rho^q\Phi=\Phi\tilde{D}^q+\sum_{j=0}^{q-1}\bar{\Delta}_\rho^{q-j}1_{m,k}V\tilde{D}^j$ for $1\leq q\leq r$.
	
	For $q=1$, $\bar{\Delta}_\rho\Phi=\Phi\tilde{D}+\bar{\Delta}_\rho 1_{m,k}V$ by Lemma \ref{D_lemma}. For $q>1$,
\[
\begin{split}
	\bar{\Delta}_\rho^q\Phi&=\bar{\Delta}_\rho\bigg(\Phi\tilde{D}^{q-1}+\sum_{j=0}^{q-2}\bar{\Delta}_\rho^{q-1-j}1_{m,k}V\tilde{D}^j\bigg)\\
				&=\Phi\tilde{D}^q+\bar{\Delta}_\rho 1_{m,k}V\tilde{D}^{q-1}+\sum_{j=0}^{q-2}\bar{\Delta}_\rho^{q-j}1_{m,k}V\tilde{D}^j\\
				&=\Phi\tilde{D}^q+\sum_{j=0}^{q-1}\bar{\Delta}_\rho^{q-j}1_{m,k}V\tilde{D}^j.
\end{split}
\]
	From the two assertions above, we get
\[
\begin{split}
	\bar{\Delta}_\rho^r\Delta^{-r}\Phi&=\bar{\Delta}_\rho^r\bigg(\Phi\tilde{C}^r-\sum_{j=0}^{r-1}\Delta^{-j}1_{m,k}V\tilde{C}^{r-j}\bigg)\\
				&=\Phi\tilde{D}^r\tilde{C}^r+\sum_{j=0}^{r-1}\bar{\Delta}_\rho^{r-j}1_{m,k}V\tilde{D}^j\tilde{C}^r-\bar{\Delta}_\rho^r\sum_{j=0}^{r-1}\Delta^{-j}1_{m,k}V\tilde{C}^{r-j}.
\end{split}
\]

\end{proof}

\subsubsection{Cancellation Between Residual Terms of $A_r\Phi_{m,k}$}\label{RTAP}
	From \eqref{eq:deci_expansion}, we can divide $A_r\Phi_{m,k}$ into two parts: $\frac{1}{\rho^r}D_\rho\Phi_{m,k}\tilde{D}^r\tilde{C}^r$ being the main term, and the rest being residual terms. In this section, we shall investigate the behavior of the residual terms.

        To facilitate the cancellation, we define an auxiliary double-sequence $\{a_{l,s}\}_{l\geq0,s\in\ZZ}$ recursively by

\[
  a_{l,s}=\left\{\begin{array}{ll}
  1&\text{if }\quad l=0,\, s\geq1\\
  0&\text{if }\quad l=0,\, s\leq0\\
  \sum_{j\leq s}a_{l-1,j}&\text{if} \quad l>0.\end{array}\right.
\]

	Let $D_\rho\bar{\Delta}_\rho^{r-j}1_{m,k}V\tilde{D}^j\tilde{C}^r=I_j^{(2)}$ and $D_\rho\bar{\Delta}_\rho^r\Delta^{-j}1_{m,k}V\tilde{C}^{r-j}=I_j^{(3)}$. We first examine the form of each $I_j^{(3)}$ before calculating the cancellation between $I_j^{(2)}$ and $I_j^{(3)}$.

\begin{lemma}
  For any $j\in\NN$ and $1\leq l\leq m$,
\[
        (\Delta^{-j}1_{m,k})_{l,s}=a_{j,l}.
\]
\end{lemma}
\begin{proof}
	First, it can easily be seen that $a_{l,s}=0$ for all $s\leq0$ by induction on $l$. Then, by definition and induction on $j$,
\[
	(\Delta^{-j}1_{m,k})_{l,s}=\sum_{n=1}^l (\Delta^{-j+1}1_{m,k})_{n,s}=\sum_{n=1}^l a_{j-1,n}=\sum_{n\leq l}a_{j-1,n}=a_{j,l}.
\]

\end{proof}

\begin{lemma}
\label{counting_balls}
  For $1\leq\kappa\leq q$ and $1\leq l\leq\eta$,
 \[
        (\bar{\Delta}_\rho^\kappa\Delta^{-q}1_{m,1})_{l\rho}=\sum_{s_1,\dots,s_\kappa=1}^\rho a_{q-\kappa,(l-\kappa)\rho+s_1+\dots+s_\kappa}.
 \]
\end{lemma}
\begin{proof}
	We shall prove this by induction on $\kappa$. For $\kappa=1$ and $l>1$,
 \[
 	(\bar{\Delta}_\rho\Delta^{-q}1_{m,1})_{l\rho}=(\Delta^{-q}1_{m,1})_{l\rho}-(\Delta^{-q}1_{m,1})_{(l-1)\rho}=a_{q,l\rho}-a_{q,(l-1)\rho}=\sum_{s_1=1}^\rho a_{q-1,(l-1)\rho+s_1}.
 \]
 
	For $l=1$,
\[
	(\bar{\Delta}_\rho\Delta^{-q}1_{m,1})_\rho=(\Delta^{-q}1_{m,1})_{\rho}=a_{q,\rho}=a_{q,\rho}-a_{q,0}=\sum_{s_1=1}^\rho a_{q-1,0+s_1}.
\]

	For $1< \kappa\leq q$ and $l>1$,
\[
\begin{split}
	(\bar{\Delta}_\rho^\kappa \Delta^{-q}1_{m,1})_{l\rho}&=(\bar{\Delta}_\rho^{\kappa-1}\Delta^{-q}1_{m,1})_{l\rho}-(\bar{\Delta}_\rho^{\kappa-1}\Delta^{-q}1_{m,1})_{(l-1)\rho}\\
			&=\sum_{s_1,\dots,s_{\kappa-1}=1}^\rho(a_{q-\kappa+1,(l-\kappa+1)\rho+s_1+\dots+s_{\kappa-1}}-a_{q-\kappa+1,(l-\kappa)\rho+s_1+\dots+s_{\kappa-1}})\\
			&=\sum_{s_1,\dots,s_\kappa}^\rho a_{q-\kappa,(l-\kappa)\rho+s_1+\dots+s_\kappa}.
\end{split}
\]

	As for $l=1$,
\[
\begin{split}
	(\bar{\Delta}_\rho^\kappa \Delta^{-q}1_{m,1})_{\rho}&=(\bar{\Delta}_\rho^{\kappa-1}\Delta^{-q}1_{m,1})_\rho\\
			&=\sum_{s_1,\dots,s_{\kappa-1}=1}^\rho a_{q-\kappa+1,(1-\kappa+1)\rho+s_1+\dots+s_{\kappa-1}}\\
			&=\sum_{s_1,\dots,s_{\kappa-1}=1}^\rho a_{q-\kappa+1,(1-\kappa+1)\rho+s_1+\dots+s_{\kappa-1}}-a_{q-\kappa+1,(0-\kappa+1)\rho+s_1+\dots+s_{\kappa-1}}\\
			&=\sum_{s_1,\dots,s_{\kappa}}^\rho a_{q-\kappa,(1-\kappa)rho+s_1+\dots+s_\kappa},
\end{split}
\]
	where the third equality follows from the fact that $s_1+\dots+s_{\kappa-1}\leq (\kappa-1)\rho$.
\end{proof}

\begin{proposition}
\label{mutual_cancel}

	For $1\leq l\leq r$,
 \[
  D_\rho\bar{\Delta}_\rho^{l}1_{m,k}V\tilde{D}^{r-l}\tilde{C}^r-D_\rho\bar{\Delta}_\rho^r\Delta^{-r+l}1_{m,k}V\tilde{C}^l=\Delta^l\bigg(1_{\eta,k}V(\tilde{D}^{r-l}\tilde{C}^{r-l}-Id)+E_{r-l}\bigg)\tilde{C}^l,
  \]
  where $E_{r-l}=\tilde{B}1_{\eta,k}V$, and $\tilde{B}$ is a diagonal matrix with $|\tilde{B}_{i,i}|\leq\rho^{r-l}$ for all $i\leq r$ and $\tilde{B}_{i,i}=0$ otherwise.
\end{proposition}
\begin{proof}

	From Lemma \ref{counting_balls}, we see that $(\bar{\Delta}_\rho^q\Delta^{-q}1_{m,1})_{l\rho}=\sum_{s_1,\dots,s_q=1}^\rho a_{0,(l-q)\rho+s_1+\dots+s_q}$. Thus, $(\bar{\Delta}_\rho^q\Delta^{-q}1_{m,1})_{l\rho}=|Z_{l,q}|$, where
\[
	Z_{l,q}=\{(s_1,\dots,s_q)\in\NN^q:\, 1\leq s_1,\dots,s_1\leq\rho,\, s_1+\dots+s_q>(q-l)\rho\}.
\]
	Note that $|Z_{l,q}|\leq \rho^q$, and $|Z_{l,q}|=\rho^q$ if $l\geq q$. Thus, $D_\rho\bar{\Delta}_\rho^q\Delta^{-q}1_{m,1}=\rho^q1_{\eta,1}-\tilde{b}$, where $\|\tilde{b}\|_\infty\leq\rho^q$ and $\tilde{b}_j=0$ for all $j\geq q$. Then, we have

\[
\begin{split}
	&\quad D_\rho\bar{\Delta}_\rho^{l}1_{m,k}V\tilde{D}^{r-l}\tilde{C}^r-D_\rho\bar{\Delta}_\rho^r\Delta^{-r+l}1_{m,k}V\tilde{C}^l\\
	&=D_\rho\bar{\Delta}_\rho^l\bigg(1_{m,k}V\tilde{D}^{r-l}\tilde{C}^{r-l}-\bar{\Delta}_\rho^{r-l}\Delta^{-(r-l)}1_{m,k}V\bigg)\tilde{C}^l\\
	&=\Delta^lD_\rho\bigg(1_{m,k}V\tilde{D}^{r-l}\tilde{C}^{r-l}-\bar{\Delta}_\rho^{r-l}\Delta^{-(r-l)}1_{m,k}V\bigg)\tilde{C}^l\\
	&=\Delta^l\bigg(1_{\eta,k}V\tilde{D}^{r-l}\tilde{C}^{r-l}-D_\rho\bar{\Delta}_\rho^{r-l}\Delta^{-(r-l)}1_{m,k}V\bigg)\tilde{C}^l\\
	&=\Delta^l\bigg(1_{\eta,k}V\tilde{D}^{r-l}\tilde{C}^{r-l}-\rho^{r-l}1_{\eta,k}V+\tilde{B}1_{\eta,k}V\bigg)\tilde{C}^l\\
	&=\Delta^l\bigg(1_{\eta,k}V(\tilde{D}^{r-l}\tilde{C}^{r-l}-\rho^{r-l}Id)+E_{r-l}\bigg)\tilde{C}^l.
\end{split}
\]

\end{proof}
\subsection{Lower Frame Bound Estimate}\label{LFBA}

	Now, we are able to answer the first question in Section \ref{proof}.
\begin{lemma}
\label{DC_LB}
	The $2$-norm of $(\frac{1}{\rho}\tilde{D}\tilde{C})^{-1}$ satisfies
	$\|(\frac{1}{\rho}\tilde{D}\tilde{C})^{-1}\|_2\leq\frac{\pi}{2}$.
\end{lemma}
\begin{proof}
	To prove the lemma, it suffices to show that for any unit-norm vector $v$, $\|\frac{1}{\rho}\tilde{D}\tilde{C}v\|_2\geq\frac{2}{\pi}$. Note that $\tilde{D}$ and $\tilde{C}$ are simultaneously diagonalizable by the hermitian matrix $B$, so for any such $v$,
\[
\begin{split}
\|\frac{1}{\rho}\tilde{D}\tilde{C}v\|_2&=\|\frac{1}{\rho}B(B^\ast\tilde{D}B) (B^\ast\tilde{C}B)B^\ast v\|_2\\
			&=\left\|diag\left(\frac{1-e^{2\pi\imath\rho \lambda_s/m}}{\rho(1-e^{2\pi\imath\lambda_s/m})}\right)(B^\ast v)\right\|_2\\
			&\geq \min_{s\in\{1,\dots,k\}}\left|\frac{1-e^{2\pi\imath\rho \lambda_s/m}}{\rho(1-e^{2\pi\imath\lambda_s/m})}\right|\\
			&=\min_s\bigg|\frac{\sin(\pi\lambda_s/\eta)}{\rho\sin(\pi\lambda_s/m)}\bigg|\geq\min_{t\in[-\eta/2,\eta/2]}\bigg|\frac{\sin(\pi t/\eta)}{\rho\sin(\pi t/m)}\bigg|\geq\frac{2}{\pi},
\end{split}
\]
	where in the second equality, we note that since $B$ is unitary, $\|MB\|_2=\|BM\|_2=\|M\|_2$ for any matrix $M$, and $\|B^\ast v\|_2=\|v\|_2=1$. The second-to-last inequality comes from the assumption that $\{\lambda_s\}_{s=1}^k\subset[-\eta/2,\eta/2]$, and the final inequality can be obtained with simple calculus, see Lemma 4.5 in \cite{KL_2018}. 

\end{proof}
\begin{lemma}[Proposition 5.2, \cite{KL_2018}]
\label{Phi_LB}
	Given the assumption in Theorem \ref{main_uni} and $n$ satisfying $n\mid m$ and $m/n\geq k$, $\Phi_{m/n,k}^\ast$ has lower frame bound larger than $\frac{m}{n}\min_s|{<}\phi_0,v_s{>}|^2=\frac{m}{n}C_{\phi_0}$.

\end{lemma}

Using Lemma \ref{DC_LB} and \ref{Phi_LB}, we are able to prove the following proposition:

\begin{proposition}
\label{lower_frame_bound}
	Suppose $\eta=m/\rho\geq k\cdot 3r$, then $A_r\Phi_{m,k}$ is a frame with lower frame bound larger than $kC_{\phi_0} (\frac{2}{\pi})^{2r}$, where $\phi_0=\sum_s c_sv_s$.
\end{proposition}
\begin{proof}
	First, note that
\[
\begin{split}
	&\quad D_\rho\bar{\Delta}_\rho^r\Delta^{-r}\Phi_{m,k}\\
	&=D_\rho\bigg[\Phi_{m,k}\tilde{D}^r\tilde{C}^r+(\bar{\Delta}_\rho^r 1_{m,k}V+\dots+\bar{\Delta}_\rho1_{m,k}V\tilde{D}^{r-1})\tilde{C}^r-\bar{\Delta}_\rho^r(1_{m,k}V\tilde{C}^r+\dots+(\Delta^{-r+1}1_{m,k}V)\tilde{C})\bigg]\\
	&=\Phi_{\eta,k}\tilde{D}^r\tilde{C}^r+D_\rho\sum_{l=1}^r(\bar{\Delta}_\rho^l1_{m,k}V\tilde{D}^{r-l}\tilde{C}^r-\bar{\Delta}_\rho^r\delta^{-r+l}1_{m,k}V\tilde{C}^l)\\
	&=\Phi_{\eta,k}\tilde{D}^r\tilde{C}^r+\sum_{l=1}^r\Delta^l D_\rho\bigg[1_{m,k}V\tilde{D}^{r-l}\tilde{C}^{r-l}-\bar{\Delta}_\rho^{r-l}\Delta^{-(r-l)}1_{m,k}V\bigg]\tilde{C}^l\\
	&=\Phi_{\eta,k}\tilde{D}^r\tilde{C}^r+\sum_{l=1}^r\Delta^l[1_{\eta,k}V(\tilde{D}^{r-l}\tilde{C}^{r-l}-I_k)]\tilde{C}^l+\Delta^l E_{r-l}\tilde{C}^{l}.
\end{split}
\]

	Now, note that $\Delta^l 1_{\eta,k}$ has nonzero entries on only the first $l$ rows. For $\Delta^l E_{r-l}$, only the first $r+l$ entries can be nonzero. Thus, the $l\cdot \lfloor\eta/k\rfloor$-th rows of $A_r\Phi_{m,k}$ is equal to the one of $\frac{1}{\rho^r}\Phi_{\eta,k}\tilde{D}^r\tilde{C}^r$. Now, the lower frame bound of $A_r\Phi_{m,k}$ is larger than the one of any of its sub-frame. In particular, its lower frame bound is larger than the one of $\frac{1}{\rho^r}\Phi_{k,k}\tilde{D}^r\tilde{C}^r$, which is $kC_{\phi_0}(\frac{2}{\pi})^{2r}$, since for any unit-norm vector $v$,
\[
	\|\frac{1}{\rho^r}\Phi_{k,k}\tilde{D}^r\tilde{C}^r v\|_2^2\geq kC_{\phi_0}\|(\frac{1}{\rho}\tilde{D}\tilde{C})^r v\|_2^2\geq kC_{\phi_0}\left(\frac{2}{\pi}\right)^{2r}.
\]

\end{proof}
\subsection{Frame Variation Bound} \label{FVB}

	In \eqref{E_est_RP}, we also need to estimate $\|(A_r\Phi_{m,k})^\ast\Delta^r\|_{\infty,2}$. To do so, we first invoke the frame variation result from \cite{KL_2018} to estimate the contribution from the main term.

\begin{lemma}[\cite{KL_2018}, Lemma 7.6]
\label{uni_var}
	For any $r\in\NN$, 
\[
	\|\Phi_{m/\rho}^\ast\Delta^r\|_{\infty,2}=\sum_{s=1}^{m/\rho}\|\Phi_{m/\rho}^\ast\Delta^r e_s\|_2\leq r2^r+\eta(2\pi\max_{1\leq j\leq k}|\lambda_j|\frac{1}{\eta})^r, 
\]

where $(e_s)_{j}=\delta(s-j)$, the $s$-th canonical coordinate.
\end{lemma}

	Now, we can estimate the $\infty$-to-$2$ norm of $(A_r\Phi_{m,k})^\ast\Delta^r$.

\begin{proposition}
\label{variation_bound}
\[
	\|(A_r\Phi_{m,k})^\ast\Delta^r\|_{\infty,2}\leq 2^{2r+2}\eta^{r-1}.
\]
\end{proposition}

\begin{proof}
	From Proposition \ref{deci_expansion} and \ref{mutual_cancel}, we see that
\[
	D_\rho\bar{\Delta}^r\Delta^{-r}\Phi_{m,k}=\Phi_{\eta,k}\tilde{D}^r\tilde{C}^r+\sum_{l=1}^{r}\Delta^l\bigg(1_{\eta,k}V(\tilde{D}^{r-l}\tilde{C}^{r-l}-\rho^{r-l}Id)+E_{r-l}\bigg)\tilde{C}^l.
\]
	Thus,
\[
\begin{split}
	&\quad \|(A_r\Phi_{m,k})^\ast\Delta^r\|_{\infty,2}=\|\frac{1}{\rho^r}(D_\rho\bar{\Delta}^r\Delta^{-r}\Phi_{m,k})^\ast\Delta^r\|_{\infty,2}\\
	&\leq\|\frac{1}{\rho^r}\tilde{D}^r\tilde{C}^r\|_2\|\Phi_{\eta,k}^\ast\Delta^r\|_{\infty,2}+2\sum_{l=1}^r\|\frac{1}{\rho^l}\tilde{C}^l\|_2\|\frac{1}{\rho^{r-l}}\tilde{D}^{r-l}\tilde{C}^{r-l}-Id\|_2\|V^\ast1_{k,\eta}\Delta^{l+r}\|_{\infty,2},
\end{split}
\]

	where we observe that $\|\frac{1}{\rho^{r-l}}E_{r-l}^\ast\Delta^{r+l}\|_{\infty,2}\leq\|V^\ast1_{k,\eta}\Delta^{r+l}\|_{\infty,2}$.
	
	Now, $\|\Phi_{\eta,k}^\ast\Delta^r\|_{\infty,2}\leq r2^r+\eta(2\pi\max_{1\leq j\leq k}|\lambda_j|\frac{1}{\eta})^r$ by Lemma \ref{uni_var}, and $\|V^\ast1_{k,\eta}\Delta^{l+r}\|_{\infty,2}=2^{l+r-1}\|\phi_0\|_2=2^{l+r-1}$. Moreover, $\|\frac{1}{\rho^r}\tilde{D}^r\tilde{C}^r\|_2\leq1$, $\|\frac{1}{\rho^{r-l}}\tilde{D}^{r-l}\tilde{C}^{r-l}-Id\|_2\leq2$, and $\|\frac{1}{\rho^l}\tilde{C}^l\|_2\leq\eta^l$. Thus,
\[
\begin{split}
	\|\frac{1}{\rho^r}(D_\rho\bar{\Delta}^r\Delta^{-1}\Phi_{m,k})^\ast\Delta^r\|_{\infty,2}&\leq r2^r+\eta(2\pi\max_{1\leq j\leq k}|\lambda_j|\frac{1}{\eta})^r+2^{r+1}\frac{(2\eta)^r-1}{2\eta-1}\\
			&\leq r2^r+\eta(2\pi\max_{1\leq j\leq k}|\lambda_j|\frac{1}{\eta})^r+2^{2r+1}\eta^{r-1}\leq 2^{2r+2}\eta^{r-1},
\end{split}
\]
	independent of $m$.
\end{proof}

\subsection{Data Storage Efficiency}\label{DSE}

	Given a mid-rise quantizer with length $2L$ and the quantized sample $q\in\CC^m$, one needs $\log(2L)$ bits to record each entry of $q$. Thus, a total of $m\log(2L)=O(\rho)$ bits is needed to fully record $q$ as $\rho\to\infty$. In this section, we shall show that with the application of adapted decimation, we may now record the decimated signal in $O(\log(\rho))$ bits, drastically fewer than originally needed.

\begin{proposition}
\label{data_storage}
	Given a mid-rise quantizer with length $2L$, it is possible to encode $D_\rho\bar{\Delta}_\rho^r\Delta^{-r}q$ with $2\eta r\log(2m)+2\eta\log(2L)$ bits in total.
\end{proposition}
\begin{proof}
	Note that for mid-rise uniform quantizers $\mathscr{A}=\mathscr{A}_0+\imath\mathscr{A}_0$ with length $2L$, each entry $q_j$ of $q$ is a number of the form 
\[
q_j=\big((2s_j+1)+\imath(2t_j+1)\big)\frac{\delta}{2},\quad-L\leq s_j,\, t_j\leq L-1.
\] 
	Then, each entry in $\Delta^{-1} q$ is the summation of at most $m$ entries in $q$, which has the form 
\[
	(\Delta^{-1} q)_j=\big((2\tilde{s}_j+\rho)+\imath(2\tilde{t}_j+\rho)\big)\frac{\delta}{2},\quad-Lm\leq \tilde{s}_j,\, \tilde{t}_j\leq (L-1)m.
\]	
	Iterating $r$ times, we see that 
\[
	(\Delta^{-r}q)_j=\big((2\tilde{s}_j+\rho)+\imath(2\tilde{t}_j+\rho)\big)\frac{\delta}{2},\quad-Lm^r\leq \tilde{s}_j,\, \tilde{t}_j\leq (L-1)m^r.
\]	
	As for $\bar{\Delta}_\rho^r \Delta^{-r}q$, we see that, for any $v\in\CC^m$, each entry of $\bar{\Delta}_\rho v$ contains at most $2$ entries of $v$. Thus,
\[
	(\bar{\Delta}_\rho^r\Delta^{-r}q)_j=\big((2\tilde{s}_j+\rho)+\imath(2\tilde{t}_j+\rho)\big)\frac{\delta}{2},\quad-L(2m)^r\leq \tilde{s}_j,\, \tilde{t}_j\leq (L-1)(2m)^r
\]
	Now, there are at most $((2L-1)(2m)^r+1)^2\leq (2L(2m)^r)^2$ choices per entry with $\eta=m/\rho$ entries in total for $D_\rho\bar{\Delta}_\rho^r\Delta^{-r}q$. Thus, it can be encoded by $\mathscr{R}=2\eta r\log(2m)+2\eta\log(2L)$ bits.

\end{proof}

\subsection{Proof of Theorem \ref{main_uni}}\label{main_proof}

\begin{proof} of Theorem \ref{main_uni}:

	By Lemma \ref{twist_scale}, 
\[
\rho^rA_rq=D_\rho\bar{\Delta}_\rho^r\Delta^{-r}q=\Delta^r D_\rho(\Delta^{-r}q).
\]
	Since $\Delta$ and $\Delta^{-1}$ are lower-triangular, we see that, for any $1\leq s\leq \eta$, there exists $\{a_j^s\}_{j=1}^s$ and $\{b_j^l\}_{j,l}$ such that
\[
	(A_rq)_s=\sum_{j=1}^s a_j^s(D_\rho\Delta^{-r}q)_j=\sum_{j=1}^sa_j^s(\Delta^{-r}q)_{j\rho}=\sum_{j=1}^sa_j^s\sum_{l=1}^{j\rho}b_l^j q_l=\sum_{\xi=1}^{s\rho}c_\xi q_\xi,
\]

	proving the first claim. The second assertion follows from Proposition \ref{lower_frame_bound}.

	Given $\Phi=\Phi_{m,k}$, $A=A_r=\frac{1}{\rho^r}D_\rho\bar{\Delta}_\rho^r\Delta^{-r}$, and $\Scal=(A\Phi)^\ast A\Phi$, the reconstruction error can be estimated as follows:

\[
\begin{split}
	\mathscr{E}=\|\Scal^{-1} (A\Phi)^\ast Aq-x\|_2&=\|\Scal^{-1}(A\Phi)^\ast A\Delta^r u\|_2\\
					&=\frac{1}{\rho^r}\|\Scal^{-1}(A\Phi)^\ast D_\rho\bar{\Delta}_\rho^r u\|_2\\
					&=\frac{1}{\rho^r}\|\Scal^{-1}(A\Phi)^\ast \Delta^r D_\rho u\|_2\\
					&\leq\frac{1}{\rho^r}\|\Scal^{-1}\|_2\|(A\Phi)^\ast\Delta^r\|_{\infty,2}\|D_\rho u\|_\infty\\
					&\leq\frac{1}{\rho^r}\big(kC_{\phi_0}(\frac{2}{\pi})^{2r}\big)^{-1}2^{2r+2}\eta^{r-1}\|u\|_\infty\\
					&=\bigg(\frac{4}{k\eta C_{\phi_0}}(\pi^2\eta)^r\bigg)\|u\|_\infty\frac{1}{\rho^r},
\end{split}
\]
	where the second inequality comes from Proposition \ref{lower_frame_bound} and Proposition \ref{variation_bound}.
		
	As for the data storage, we see from Proposition \ref{data_storage} that one can encode the data $A_rq$ with $\mathscr{R}=2\eta r\log(2m)+2\eta\log(2L)$ bits in total. 
	
	Note that
\[
	e^{\frac{-\mathscr{R}}{2\eta}}=(2m)^{-r}\cdot\frac{1}{2L}=\frac{1}{2L}\left(\frac{\eta}{2}\right)^r\cdot \frac{1}{\rho^r}.
\]
	
	Thus, as the function of bits used, the reconstruction error satisfies
\[
\begin{split}
	\mathscr{E}(\mathscr{R})&\leq\bigg(\frac{4}{k\eta C_{\phi_0}}(\pi^2\eta)^r\bigg)\|u\|_\infty\frac{1}{\rho^r}\\
		&= C_{k,\eta,\phi_0,L}\|u\|_\infty \frac{1}{2L}\left(\frac{\eta}{2}\right)^r \frac{1}{\rho^r}\\
		&= C_{k,\eta,\phi_0,L}\|u\|_\infty e^{\frac{-\mathscr{R}}{2\eta}},
\end{split}
\]
	where $C_{k,\eta,\phi_0,L}=\frac{8L}{k\eta C_{\phi_0}}(2\pi^2)^r$.

\end{proof}

\section{Acknowledgement}
The author would like to thank the support from ARO Grant W911NF-17-1-0014, NSF-DMS Grant 1814253, and J. Benedetto for all the thoughtful advice and insights. Further, the author appreciates the constructive analysis and suggestions of the referees.


\bibliographystyle{amsplain}

\bibliography{ref}

\end{document}